\documentclass[10pt, a4paper, reqno]{amsart}
\usepackage{amscd, amsmath, amstext, amsthm, amssymb, amsopn, amssymb, amsxtra, amsfonts}

\usepackage{tikz}
\usepackage{graphicx}
\usepackage{fancyhdr}

\usepackage{paralist}
\usepackage[mathcal]{euscript}
\usepackage[english]{babel}
\usepackage[latin1]{inputenc}

\newtheorem{thm}{Theorem}
\newtheorem{lem}[thm]{Lemma}
\newtheorem{prop}[thm]{Proposition}

\makeatletter
\renewcommand*\@maketitle{%
  \normalfont\normalsize
  \@adminfootnotes
  \@mkboth{\@nx\shortauthors}{\@nx\shorttitle}%
  \global\topskip42\p@\relax 
  \@settitle
  \ifx\@empty\authors \else \@setauthors \fi
  \ifx\@empty\@date \else {\vskip 1em \vtop{\centering\large\@date\@@par}}\fi
  \ifx\@empty\@dedicatory
  \else
    \baselineskip18\p@
    \vtop{\centering{\footnotesize\itshape\@dedicatory\@@par}%
      \global\dimen@i\prevdepth}\prevdepth\dimen@i
  \fi
  \@setabstract
  \normalsize
  \if@titlepage
    \newpage
  \else
    \dimen@34\p@ \advance\dimen@-\baselineskip
    \vskip\dimen@\relax
  \fi
} 
\renewcommand*\@adminfootnotes{%
  \let\@makefnmark\relax  \let\@thefnmark\relax
  \ifx\@empty\@subjclass\else \@footnotetext{\@setsubjclass}\fi
  \ifx\@empty\@keywords\else \@footnotetext{\@setkeywords}\fi
  \ifx\@empty\thankses\else \@footnotetext{%
    \def\par{\let\par\@par}\@setthanks}%
  \fi
}
\makeatother

\begin{document}
\bibliographystyle{amsplain}

\title{Asymptotic bounds on the\\equilateral dimension of hypercubes\vspace{-30pt}}
\author{~}

\date{July 29, 2014}


\renewcommand{\thefootnote}{}
\hspace{-1000pt}\footnote{\hspace{5.5pt}2010 \emph{Mathematics Subject Classification}: Primary 05C12; Secondary 11H71, 51B20.}
\hspace{-1000pt}\footnote{\hspace{5.5pt}\emph{Key words and phrases}: hypercube, equilateral dimension.\vspace{1.6pt}}

\hspace{-1000pt}\footnote{$^{\text{a}}$\,Computer Science Division, University of California, Berkeley CA 94720-1776, USA\newline\phantom{x}\hspace{13pt}E-Mail: lorenz.minder@a3.epfl.ch\vspace{1pt}}

\hspace{-1000pt}\footnote{$^{\text{b}}$\,University of Cambridge, Computer Lab., JJ Thomson Avenue, Cambridge CB3 0FD, UK\newline\phantom{x}\hspace{13pt}E-Mail: thomas.sauerwald@cl.cam.ac.uk\vspace{1pt}}

\hspace{-1000pt}\footnote{$^{\text{c}}$\,Universit\"at Wuppertal, FB C--Mathematik, Gau\ss{}stra\ss{}e 20, D-42119 Wuppertal, GERMANY\newline\phantom{x}\hspace{13pt}E-Mail: wegner@math.uni-wuppertal.de}


\begin{abstract}\vspace{25pt}
A subset of the finite dimensional hypercube is said to be equilateral if the distance of any two distinct points equals a fixed value. The equilateral dimension of the hypercube is defined as the maximal size of its equilateral subsets. We study asymptotic bounds on the latter quantity considered as a function of two variables, namely dimension and distance.
\end{abstract}

\maketitle

\begin{picture}(0,0)
\put(43.5,105){{\sc\MakeLowercase{Lorenz Minder}}\hspace{0.8pt}$^{\text{a}}$, {\sc\MakeLowercase{Thomas Sauerwald}}\hspace{0.8pt}$^{\text{b}}$ {\sc and} {\sc\MakeLowercase{Sven-Ake Wegner}}\hspace{0.8pt}$^{\text{c}}$}

\end{picture}

\vspace{-20pt}

\section{Introduction}\label{Intro}

The notion of equilateral sets, i.e., sets in which the distance of any two distinct points equals a fixed value, can be defined in broad generality -- that is in arbitrary metric spaces.
\vspace{3pt}
\\In 1983, Kusner \cite{Olla-Podrida} raised the question of determining equilateral sets in the finite dimensional $\ell^p$-spaces for $1\leqslant p\leqslant\infty$. He computed the so-called equilateral dimension, that is the maximal size of equilateral sets, of the Hilbert space $\ell^2(n)$ and of $\ell^{\infty}(n)$. For the other cases of $p$ he formulated conjectures which are very persuasive but turned out to constitute surprisingly hard and interdisciplinary problems. In particular, the case of $\ell^1(n)$ appears to be a fairly easy exercise at a first glance but is resisting a complete solution for more than thirty years now. However, during this time many results, using techniques from various areas of mathematics such as functional analysis, probability theory, combinatorics, approximation theory and algebraic topology, have been obtained by several authors, e.g., Alon, Pudl\'{a}k \cite{AlonPudlak}, Bandelt et al.~\cite{BandeltChepoiLaurent} and Koolen et al.~\cite{KoolenLaurentSchrijver}. We 
refer 
to the survey \cite{Swanepoel2004} of Swanepoel for detailed references, an overview of the state of the art concerning equilateral sets in normed spaces and historical comments.
\vspace{3pt}
\\A completely different area where equilateral sets can be studied are finite, undirected and connected graphs with the usual shortest-path metric. Being a structural invariant, it is a natural objective to compute the equilateral dimension of certain graph classes. In addition, equilateral sets might be of use in the context of information dissemination (similar to Feige et al.~\cite{FPRU90}), i.e., the problem of spreading a message held by a set of source nodes to a set of target nodes (broadcasting), by using an algorithm which regulates the communications in the neighborhood of any point. 
\vspace{3pt}
\\A third area in which equilateral sets occur naturally is the theory of codes. Equilateral subsets of the hypercube are nothing but equidistant codes in the language of coding theory (cf.~the books of MacWilliams, Sloane \cite{MacWilliamsSloane1977} or Huffman, Pless \cite{HuffmanPless2003}) and constitute a classical research topic with important applications for instance to error correcting codes. They are closely related to constant weight codes and to the theory of 2-designs. We refer to Bogdanova et al.~\cite{BogdanovaetAl} for detailed definitions and results explaining the connections between the latter objects and their theories. In addition we refer e.g.~to Brower et al.~\cite{BSSS1990}, Heng, Cooke \cite{HC98} or Fu et al.~\cite{FKLW03} for recent results on equidistant and constant weight codes.
\vspace{3pt}
\\In this article we study equilateral sets in the hypercube. After introducing our notation in Section \ref{Notation}, we perform an asymptotic analysis of the equilateral dimension in Sections \ref{sec:linear}--\ref{sec:poynomial}. In Section \ref{sec:linear}, we assume a linear relation between distance and dimension. In Section \ref{sec:constant} we consider a constant distance, but increasing dimensions and in the final Section \ref{sec:poynomial}, we study the case of a polynomial relation between the latter quantities. In all three sections we analyze the maximal asymptotic growth of the equilateral dimension.
\vspace{3pt}
\\Complementing prior work (performed in the language of codes or designs), which focused on computing the equilateral dimension for fixed distance and dimension exactly, our asymptotic analysis yields more intuitive and less technical results. It may be helpful for applications where dimension and distance are not fixed but coupled by a certain functional relation. Putting it differently, we relax the problem of constructing equilateral sets for fixed parameters which turns out to give new insights on their asymptotic distribution. In particular, we establish that for most reasonable functional relations between distance and dimension, the growth of the equilateral dimension of the $n$-dimensional hypercube is at least $n^{2/3}$.

\section{Notation}\label{Notation}

For $n\geqslant1$, the $n$-dimensional hypercube $Q_n$ is a finite, connected and undirected graph, which is defined as follows. Its vertices are all binary sequences $x=(x_i)_{i=1,\dots,n}$ of length $n$. Two vertices $x$ and $y$ are connected by an edge if the corresponding sequences differ in exactly one entry, i.e., if and only if there exists exactly one index $i$ such that $x_i\not= y_i$.
\vspace{-9pt}
\\We endow $Q_n$ with the shortest-path metric $\rho$ and note that the latter coincides with the so-called Hamming distance, i.e., $\rho(x,y)=|\{i\:;\:x_i\not=y_i\}|$.
\smallskip
\\Let $d$ be a positive integer. A non-empty subset $S\subseteq Q_n$ is said to be $d$-equilateral, if $\rho(x,y)=d$ for all $x,y\in S$ with $x\not=y$. By $e_d(n)\equiv e_d(Q_n)$ we denote the maximal size of $d$-equilateral sets in $Q_n$. $e_d(n)$ is referred to as the equilateral dimension of $Q_n$ with respect to the metric $\rho$.
\smallskip
\\Our goal is the asymptotic analysis of $e_d(n)$ considered as a function of two variables. Here, we let $n$ go to infinity and we may or may not assume the same for $d$. More precisely, we would like to know the value, if it exists, of $\lim_{k\rightarrow\infty}e_{d_k}(n_k)$ for any sequence $(d_k,n_k)_{k\in\mathbb{N}}\subseteq\mathbb{N}\times\mathbb{N}$ with $\lim_{k\rightarrow\infty}n_k=\infty$ and $d_k\leqslant n_k$ may be bounded or unbounded (note that $e_d(n)=1$ for $d>n$). 
Since solving this problem in its full generality is out of reach, we restrict ourselves to sequences of special types. In Section \ref{sec:linear}, we assume a linear relation between $n$ and $d$. In Section \ref{sec:constant} we consider a constant $d$, but increasing $n$ and in the final Section \ref{sec:poynomial}, we study the case of a polynomial relation between $n$ and $d$. However, many of the corresponding sequences might not converge, since for instance  $e_d(n)=2$, if $d\leqslant n$ is odd, as can be checked by straightforward computations. To avoid such \textquotedblleft{}pathologies\textquotedblright{} and in order to exhibit indeed the maximal growth of the equilateral dimension, we will use customized \textquotedblleft{}measure functions\textquotedblright{} that form a limit superior and quantify the growth in all of the three settings mentioned above.

\section{Asymptotics I: Linear Relation}\label{sec:linear}

We start our asymptotic analysis of $e_d(n)$ with the very natural case that the ratio of distance and dimension is fixed to some given value $\gamma\in\,]0,1]$, i.e., we study sequences with $d_k=\lfloor \gamma n_k \rfloor$. Note that for a given $\gamma$ all sequences $(\lfloor \gamma{}n_k \rfloor,\,n_k)_{k\in\mathbb{N}}$ are subsequences of $(\lfloor \gamma n \rfloor,\,n)_{n\in\mathbb{N}}$. For some aspects it is thus enough to study only the single sequence $(e_{\lfloor \gamma n \rfloor}(n))_{n\in\mathbb{N}}$.
\smallskip
\\The first question to ask is whether $(e_{\lfloor \gamma n \rfloor}(n))_{n\in\mathbb{N}}$ is bounded or not. Using the Plotkin bound (e.g., \cite[Theorem 2.2.1 and 2.2.4]{MacWilliamsSloane1977}) and the fact that $e_d(n)=2$, if $d\leqslant n$ is odd, $e_{\lfloor \gamma n \rfloor}(n)$ is at most $\max \{ 2, 2 \lfloor  \gamma/(2 \gamma - 1) \rfloor \} $ if $\gamma > 1/2$ and thus bounded. On the other hand, Theorem \ref{g-Prop} will show that for each $0 < \gamma \leqslant 1/2$ there exist sequences $(n_k)_{k \in \mathbb{N}}$ such that $(e_{\lfloor \gamma n_k \rfloor}(n_k))_{k \in \mathbb{N}}$ is unbounded; we define with foresight
$$
g(\gamma):=\limsup_{n\rightarrow\infty}{\textstyle\frac{e_{\lfloor \gamma n \rfloor}(n)}{n}} \; \text{ for } \: \gamma \in\;]0,1].
$$
Note that $g(\gamma) \leqslant 1$, since $e_{\lfloor \gamma n \rfloor}(n) \leqslant n+1$ which follows from Kusner's result (the equilateral dimension of $\ell^2(n)$ equals $n+1$, cf.~\cite{Olla-Podrida}) by embedding $Q_n$ into $\ell^2(n)$ or, for instance, by considering an equilateral set as a 2-design and applying Fisher's inequality, cf.~\cite[p.~63]{MacWilliamsSloane1977}. We have the following results on $g$.

\begin{thm}\label{g-Prop} Let $0<\gamma\leqslant1$ and $g$ be defined as above.
\begin{itemize}
\item[(1)] For $0<\gamma\leqslant1/2$ we have $2 \gamma \leqslant g(\gamma)\leqslant1$.
\item[(2)] For $1/2<\gamma\leqslant1$ we have $g(\gamma)=0$.
\end{itemize}
\end{thm}
\begin{proof}(1) Let $n_k:=\lceil 2^k / \gamma \rceil$. Then $\lfloor \gamma n_k \rfloor = 2^k$ for any $\gamma \leqslant 1/2$. Applying a result of Bose, Shrikhande \cite{BoeseShrikhande} (see \cite[Lemma 5]{BogdanovaetAl}) yields $e_{2^k}(n_k) \geqslant 2^{k+1} = 2 \lfloor \gamma n_k \rfloor$, which in turn implies $e_{\lfloor \gamma n_k \rfloor}(n_k) / n_k \geqslant (2 \lfloor \gamma n_k \rfloor) / n_k \geqslant 2 \gamma - 2 / n_k \geqslant 2 \gamma - 2/(2^k/\gamma) \geqslant 2 \gamma - 2^{-k}$. Therefore, $g(\gamma) \geqslant 2 \gamma$.
\smallskip
\\(2) As we observed above, $e_{\lfloor \gamma{}k \rfloor}(k)$ is bounded for $\gamma>1/2$, which immediately implies $g(\gamma)=0$ in this case.
\end{proof}

\section{Asymptotics II: Constant Dimension}\label{sec:constant}

Since the definition of $g$ in Section \ref{sec:linear} cannot be extended to $\gamma=0$, it is natural to consider sequences $(d_k,n_k)_{k \in \mathbb{N}}$ with $\lim_{k \rightarrow \infty} d_k / n_k=0$; the easiest sequences with this property are those with a constant dimension $d_k=d$. In order to measure the growth of $e_d(n)$ we define
$$
h(d):=\limsup_{n\rightarrow\infty}{\textstyle\frac{e_d(n)}{n}} \: \text{ for } \: d\in\mathbb{N}.
$$
Note that the latter is in fact the limit of the sequence $(e_d(n)/n)_{n\in\mathbb{N}}$ as the proof of Theorem \ref{constant-d-prop} will show. Therefore, a consideration of subsequences $(d,\,n_k)_{k\in\mathbb{N}}$ is dispensable. We obtain the following results on $h$ using results of Deza \cite{Deza1973} and van Lint \cite{vanLint1973}.

\begin{thm}\label{constant-d-prop} Let $d$ be a positive integer and $h$ be defined as above. If $d$ is odd, then $h(d)=0$. If $d$ is even, then $h(d)=2/d$.
\end{thm}
\begin{proof} If $d$ is odd then $e_d(n)/n\leqslant 2/n\rightarrow0$ for $n\rightarrow\infty$ holds. Let $d$ be even. We claim that there exists $N$ such that $e_d(n)=\lfloor{}2n/d\rfloor{}$ for $n\geqslant{}N$ holds. We adopt the notation of \cite[Section 1]{vanLint1973}: Let us say that a $d$-equilateral subset of $Q_n$ of size $m$ is trivial, if in the associated $m\times{}n$-matrix each column has $m$ or $m-1$ equal entries. We choose $k$ such that $d=2k$ and we select $N$ such that $\lfloor{}N/k\rfloor{}>k^2+k+2$ holds. By \cite[Th\'{e}or\`{e}me 1.1.(ii)$'$ and (iii)]{Deza1973} for every $n \geqslant N$ all $2k$-equilateral subsets of $Q_n$ of size at least $\lfloor n/k \rfloor$ are trivial. Assume that there is a trivial $2k$-equilateral set of size $ \lfloor{}n/k\rfloor{}+1$. From the definition we infer that in every row exactly $k$ entries are equal to one. Since in each column at most one entry is one, we obtain a contradiction, as $k  (\lfloor{}n/k\rfloor{}+1) > k (n/k)=n$. 
A trivial $2k$-equilateral subset of $Q_n$ of size $\lfloor{}n/k\rfloor{}$ can be constructed in straightforward manner, which establishes the claim and implies $e_d(n)/n=\lfloor{}2n/d\rfloor{}/n$ and thus $2/d-1/n=(2n/d-1)/n\leqslant e_d(n)/n\leqslant 2/d$ for all $n$. Thus, $e_d(n)/n\rightarrow2/d$ for $n\rightarrow\infty$ follows.
\end{proof}

\section{Asymptotics III: Polynomial Relation}\label{sec:poynomial}

The results presented so far suggest to study the case where $d_k$ is not fixed, but also does not grow linearly in $n_k$. A natural choice to consider are sequences $(d_k,\,n_k)_{k\in\mathbb{N}}$ where $d_k = \lfloor{}n_k^{\alpha}\rfloor{}$ for $0<\alpha<1$.
\smallskip
\\In the setting of linear distances we obtained that the equilateral dimension either grows linearly, or is bounded. Thus, it is natural to presume that in the current setting $e_{n^{\alpha}}(n)$ may behave asymptotically like $n^{\beta}$ for some $\beta=\beta(\alpha)$. Therefore we define
$$
f(\alpha) := \limsup_{n \rightarrow \infty} \log_{n}(e_{\lfloor n^{\alpha} \rfloor}(n)) \: \text{ for } \alpha \in \;]0,1[. 
$$
Let us mention that this definition could be extended to the case $\alpha=0$, but then $f(0)=\limsup_{n \rightarrow \infty} \log_{n}(e_{1}(n))=0$ and we are again in the situation of a constant dimension. Moreover, for $\alpha=1$ we would also obtain $f(\alpha)=0$ and return to the case of a linear relation.
\smallskip
\\For the proof of our final Theorem \ref{f(alpha)-thm} we need the following result, which can be proved in an elementary way using ideas of Deza \cite{Deza1973} and van Lint \cite{vanLint1973}. For the sake of completeness we sketch the proof.

\begin{prop}\label{FF-5}
Let $q$ be any prime power and $k$ be an arbitrary integer. Then we have $e_{d(q,k)}(n(q,k)) \geqslant s(q,k)$, where $n(q,k) = q/(q-1) \cdot (1-q^{-2^{k+1}}) \cdot q^{2^{k+1}}$, $d(q,k) = 2 q ^{2^{k+1}-1}$ and $s(q,k) = q^{2^{k+1}}$.
\end{prop}
\begin{proof} (1) We first show the following: Let $q_0$ be a prime power and assume that there is a $d_0$-equilateral subset $R$ of $Q_{n_0}$ of size $q_0$. Put $n=n_0(q_0+1)$, $d=d_0q_0$ and $q=q_0^2$. Then $e_d(n)\geqslant q$ holds.
\smallskip
\\Let in the sequel $i$, $j$, $\mu$ and $\nu$ denote non-negative integers. We enumerate the elements of $\mathbb{F}_{q_0}$ by $x_0,\dots,x_{q_0-1}$ and those of $R$ by $r_{x_0},\dots,r_{x_{q_0-1}}$. For $i=0,\dots,q_0-1$ we regard each $r_{x_i}$ as a row vector $[r_{x_i,0}\;\cdots\;r_{x_{i},n_0-1}]$. We define the following two matrices for $x\in\mathbb{F}_{q_0}$. $A_x$ is the $q_0\times{}n_0$-matrix whose rows are all copies of $r_x$ and $B_x$ is the matrix whose $i$-th row is equal to $r_{x_i-x}$ where we number the rows by $i=0,\dots,q_0-1$ and the columns by $j=0,\dots,n_0-1$. We define matrices

{\small
$$
A=\begin{bmatrix}
A_{x_0}      \\
A_{x_1}      \\
\vdots       \\
A_{x_{q_0-1}}\\
\end{bmatrix}
\;\;\text{\normalsize and \small}\;\;
B=\begin{bmatrix}
B_{x_0x_0}       & B_{x_0x_1}       & \cdots & B_{x_0x_{q_0-1}}      \\
B_{x_1x_0}       & B_{x_1x_1}       & \cdots & B_{x_1x_{q_0-1}}      \\
\vdots           & \vdots           &        & \vdots                \\
B_{x_{q_0-1}x_0} & B_{x_{q_0-1}x_1} & \cdots & B_{x_{q_0-1}x_{q_0-1}}\\
\end{bmatrix}
$$}

\noindent{}where we refer to $A_{x_{\mu}}$ as the $\mu$-th row-block of $A$ and to $[B_{x_{\mu}x_{0}}\,\cdots\,B_{x_{\mu}x_{q_0-1}}]$ as the $\mu$-th row-block of $B$. The row-blocks are enumerated by $\mu=0,\dots,q_0-1$ (for $A$ and $B$) and the columns of $B$ are enumerated by $\nu=0,\dots,q_0-1$. To end our construction, we put $S=[A\;B]$ which is a $q\times{}n$-matrix. The rows of $S$ form an $d$-equilateral subset of $Q_n$. In order to show the latter, let $s_{(\mu,i)}$ be the $i$-th row of the $\mu$-th row-block of $S$. By case analysis, it follows $\rho(s_{(\mu_1,i_1)},s_{(\mu_2,i_2)})=d_0q_0(1-\delta_{\mu_1,\mu_2}\delta_{i_1,i_2})$ which shows $\rho(s_{(\mu_1,i_1)},s_{(\mu_2,i_2)})=d_0q_0=d$ for $(\mu_1,i_1)\not=(\mu_2,i_2)$ as desired.
\smallskip
\\(2) As a direct consequence of (1) we get $e_{2q}(q(q+1))\geqslant q^2=n-q$ for each prime power $q$.
\smallskip
\\(3) Applying (1) to (2) we obtain an equilateral set whose size is again a prime power. Iterating (1) $k$ times yields the desired result.
\end{proof}

To complete our preparations for the proof of Theorem \ref{f(alpha)-thm}, we observe the following two simple inequalities which we will frequently use later on.

\begin{lem}\label{lem:simple} Let $d$, $d_1$, $d_2$ as well as $n$, $n_1$ and $n_2$ be positive integers.
\begin{itemize}
\item[(1)] We have $e_{d}(n+1) \geqslant e_{d}(n)$.
\item[(2)] We have $e_{d_1+d_2}(n_1+n_2) \geqslant\min\{ e_{d_1}(n_1),\,e_{d_2}(n_2)\}$.\qed
\end{itemize}
\end{lem}

We now present our results on $f$ (cf. Figure 1 for an illustration).

\begin{thm}\label{f(alpha)-thm} Let $f$ be defined as above. If $\alpha \in\;]0,1/2]$ we have $f(\alpha) = \max \{ 1-\alpha,2 \alpha \}$. If $\alpha \in [1 - 2^{-k}, 1- 2^{-(k+1)}]$ holds for some integer $k \geqslant 1$ we have $f(\alpha) \geqslant \max \{ (1-\alpha)/2^{-k}, \alpha/(1 - 2^{-(k+1)}) \}$. In particular, $\min_{\alpha \in\:]0,1[} f(\alpha) = 2/3$ holds.
\end{thm}
\begin{proof} Consider first the case $0 < \alpha \leqslant 1/2$. From \cite[Th\'{e}or\`{e}me 1.1]{Deza1973} we get the estimate $t_d(n)\leqslant{}e_d(n)\leqslant{}\max\{(d/2)^2+d/2+2,\,t_d(n)\}$, where $t_d(n)$ is the cardinality of the largest trivial equilateral set, see the proof of Theorem \ref{constant-d-prop}. There we showed that $t_d(n)=\lfloor{}2n/d\rfloor{}$ holds for $2n\geqslant{}d$ and $d$ even. Put $d=\lfloor n^{\alpha} \rfloor$ to get $1/C\cdot{}n^{1-\alpha}\leqslant e_{\lfloor n^{\alpha} \rfloor}(n)\leqslant C  \max\{ n^{2\alpha},\,n^{1-\alpha}\}$ for some $C\geqslant1$; taking logarithms yields $(1-\alpha) - \log_{n} C  \leqslant \log_n e_{\lfloor n^{\alpha} \rfloor}(n)\leqslant\max\{2\alpha,\,1-\alpha\} + \log_{n} C$ for large $n$ where $\lfloor n^{\alpha} \rfloor$ is even (since $\alpha<1$ there are infinitely many $n$ for which this holds). Thus $1-\alpha\!\leqslant{}\!f(\alpha)\leqslant\!\max\{2\alpha,\,1-\alpha\}$.
\smallskip
\\Proposition \ref{FF-5} with $k=0$ yields $e_{2q}(q(q+1))\geqslant q^2$ for any prime power $q$. Using Lemma \ref{lem:simple}.(1) we infer that $e_{2q}(4q^2)\geqslant q^2$ holds for any prime power $q$. Denote by $\mathcal{Q}$ the set of all prime powers and consider the sequence $(n_q)_{q\in\mathcal{Q}}$ defined by $n_q:=(2q)^{1/\alpha}$, i.e., $2q=n_q^\alpha$, $4q^2=n_q^{2\alpha}$ and $q^2=n_q^{2\alpha}/4$. We get $e_{n_q^\alpha}(n_q^{2\alpha})\geqslant n_q^{2\alpha}/4$ for all $q\in\mathcal{Q}$ and put $m_q:=\lceil n_q \rceil$. Then $e_{n_q^{\alpha}}(m_q) \geqslant e_{n_q^{\alpha}}(n_q^{2 \alpha}) \geqslant n_q^{2 \alpha}/4$ follows from Lemma \ref{lem:simple}.(1). We claim that $n_q^{\alpha}= \lfloor m_q^{\alpha} \rfloor = \lfloor \lceil n_q \rceil^{\alpha} \rfloor$. This is obvious if $n_q$ is an integer, since $n_q^{\alpha}$ is also an integer. If $n_q$ is not an integer, we find $0 < \varepsilon < 1$ such that $\lfloor \lceil n_q \rceil^{\alpha} \rfloor = \lfloor (n_q+\varepsilon) ^{\alpha} \rfloor
\leqslant \lfloor n_q^{\alpha} + \varepsilon^{\alpha} \rfloor    = n_q^{\alpha}$, where the last equality follows from $\varepsilon^{\alpha} < 1$. On the other hand, we have $\lceil n_q \rceil^{\alpha} \geqslant n_q^{\alpha}$ and since $n_q^{\alpha}$ is an integer, we obtain $ \lfloor \lceil n_q \rceil^{\alpha}  \rfloor \geqslant n_q^{\alpha}$. This proves $e_{\lfloor m_q^{\alpha} \rfloor}(m_q) \geqslant 1/C \cdot m_q^{2 \alpha}$ for sufficiently large $q$, where $C \geqslant 1$. Taking logarithms, we obtain $f(\alpha) \geqslant 2 \alpha$ and recalling our previous bounds on $f(\alpha)$ we conclude $f(\alpha)=\max(2\alpha,1-\alpha)$.
\smallskip
\\Before we investigate the case $1/2 < \alpha < 1$, we study $f(\beta)$ where $\beta=1-2^{-(k+1)}$ for some integer $k\geqslant1$. Using Proposition \ref{FF-5}  we obtain $(n_q)_{q\in\mathcal{Q}}$ with $n_q \equiv n(q,k) \geqslant  q^{2^{k+1}}$ for sufficiently large $q$ and $ d_q \equiv d(q,k)  = 2 q^{2^{k+1}-1}$ such that $e_{d_q}(n_q) \geqslant s_q$ with $s_q \equiv s(q,k) = q^{2^{k+1}}$. We also have that $n_q^{\beta}=n_q^{1-2^{-(k+1)}}\leqslant (3/2)^{1-2^{-(k+1)}} q^{2^{k+1}-1} \leqslant 3/4 \cdot d(q,k) \leqslant d(q,k) \equiv d_q$. We put $m_q := \lceil d_q^{1/\beta} \rceil$, which along with $m_q \geqslant n_q$ implies $e_{d_q}(m_q) \geqslant e_{d_q}(n_q) \geqslant s_q \geqslant 2/3\cdot n_q \geqslant 1/C\cdot d_q^{1/\beta} \geqslant 1/C\cdot m_q$. Moreover, $\lfloor m_q^{\beta} \rfloor = \lfloor \lceil d_q^{1/\beta} \rceil^{\beta} \rfloor = d_q$, which is shown exactly as in the last paragraph. Taking logarithms gives $f(\beta)=1$.
\smallskip
\\Let $\alpha\in\,]1-2^{-k},1-2^{-(k+1)}[$. Recall that for $\beta=1-2^{-k}$, there exists a sequence $(n_q)_{q\in\mathcal{Q}}$ and a constant $C \geqslant 1$ with $e_{d_q}(n_q) \geqslant 1/C  \cdot n_q$ with $d_q=\lfloor n_q^{\beta} \rfloor$, provided that $q$ is sufficiently large. Define $d_q' := d_q \cdot 2 \lfloor n_q^{(\beta-\alpha)/(\alpha-1)   } \rfloor$ and $n_q' := n_q \cdot 2 \lfloor n_q^{(\beta-\alpha)/(\alpha-1)   } \rfloor$. We prove that $(n_q')^{\alpha} \leqslant d_q'$. Since $\alpha < 1$ is fixed, we can choose $q$ large enough such that $2 \cdot \lfloor n_q^{\beta} \rfloor \cdot \lfloor n_q^{(\beta-\alpha)/(\alpha-1)} \rfloor \geqslant 2^{\alpha} \cdot n_q^{\beta} \cdot n_q^{(\beta-\alpha)/(\alpha-1)}$. We estimate
$$
(n_q \cdot 2 \lfloor n_q^{\frac{\beta-\alpha}{\alpha-1}} \rfloor)^{\alpha}  \leqslant n_q^{\alpha} \cdot 2^{\alpha} \cdot n_q^{\alpha\frac{\beta-\alpha}{\alpha-1}} = 2^{\alpha}\cdot n_q^{\beta} \cdot n_q^{\frac{\beta-\alpha}{\alpha-1}}   \leqslant 2 \cdot \lfloor n_q^{\beta}\rfloor\cdot\lfloor n_q^{\frac{\beta-\alpha}{\alpha-1}}\rfloor = d_q'
$$
and apply Lemma \ref{lem:simple}.(2) to obtain $e_{d_q'}(n_q') \geqslant e_{d_q}(n_q) \geqslant 1/C \cdot n_q$. We put $m_q := \lceil (d_q')^{1/\alpha} \rceil$. Then $d_q' = \lfloor m_q^{\alpha} \rfloor = \lfloor  \lceil (d_q')^{1/\alpha} \rceil^{\alpha} \rfloor$ holds, which can be shown as above. Since $m_q \geqslant n_q$, we may apply Lemma \ref{lem:simple}.(1) to obtain $e_{\lfloor m_q^{\alpha}\rfloor}(m_q) \geqslant 1/C \cdot n_q$. Moreover,
$$
m_q = \lceil (d_q')^{1/\alpha} \rceil  \leqslant 2 ( d_q \cdot 2 \lfloor n^{\frac{\beta-\alpha}{\alpha-1}} \rfloor)^{1/\alpha}\leqslant 2^{1+1/\alpha} \cdot n_q^{\frac{\beta(\alpha-1)+\beta-\alpha}{\alpha(\alpha-1)}},
$$
which implies $n_q \geqslant 1/C \cdot m_q^{\alpha(\alpha-1)/(\beta(\alpha-1)+\beta-\alpha)}=1/C\cdot m_q^{(\alpha-1)/(\beta-1)}$ from which we conclude $e_{\lfloor m_q^{\alpha}\rfloor}(m_q) \geqslant 1/C \cdot m_q^{(\alpha-1)/(\beta-1)}$. Taking logarithms we obtain $f(\beta) \geqslant (\alpha-1)/(\beta-1)$.
\smallskip
\\For the same $\alpha$, we put $\beta=1-2^{-(k+1)}$. Again, there exists a sequence $(n_q)_{q\in\mathcal{Q}}$ and a constant $C \geqslant 1$ such that $e_{d_q}(n_q) \geqslant 1/C \cdot n_q$ holds for sufficiently large $q$ where $d_q=\lfloor n_q^{\beta} \rfloor$.
Define $m_q:=\lceil \lfloor n_q^{\beta} \rfloor^{1/\alpha} \rceil$ to get $d_q = \lfloor m_q^{\alpha} \rfloor=\lfloor \lceil d_q^{1/\alpha} \rceil^{\alpha} \rfloor$ with the same arguments as above. Moreover, $m_q =  \lceil \lfloor n_q^{\beta} \rfloor^{1/\alpha} \rceil \geqslant \lceil 2^{-1/\alpha} n_q^{\beta/\alpha} \rceil \geqslant n_q$ for sufficiently large $q$, since $\beta > \alpha$. Hence we may apply Lemma \ref{lem:simple}.(1) a last time to obtain $e_{d_q}(m_q) \geqslant e_{d_q}(n_q) \geqslant 1/C \cdot m_q^{\alpha/\beta}$. Taking logarithms, we arrive at $f(\alpha) \geqslant \alpha / \beta$. Therefore, $f(\alpha)\geqslant \max((1-\alpha)/ 2^{-k},\alpha/(1-2^{-(k+1)}))$ holds whenever $\alpha\in [1-2^{-k},1-2^{-(k+1)}]$ for some integer $k\geqslant1$.
\smallskip
\\From the formulas resp.~estimates which we just established it follows directly that $\min_{\alpha\in\,]0,1/2]} f(\alpha)=f(1/3)=2/3$ and $\min_{\alpha\in[1-2^{-k},1-2^{-(k+1)}]} f(\alpha)=f((2^{k+1}-1)/(2^{k+1}+1))=2^{k+1}/(2^{k+1}+1)$ holds for any $k\geqslant1$. Since $(2^{k+1})/(2^{k+1}+1) \geqslant 2/3$ for every $k \geqslant 1$, we obtain $\min_{\alpha \in\;]0,1[} f(\alpha) = 2/3$.
\end{proof}

A simple inspection of the first part of the proof of Theorem \ref{f(alpha)-thm} shows that also for
more slowly growing functions, e.g., $d(n)=\lfloor \log n \rfloor$, we have $\limsup_{n \rightarrow \infty} \log_{n}(e_{\lfloor d(n) \rfloor}(n)) \geqslant \limsup_{n \rightarrow \infty} \log_{n}(\lfloor 2n/\lfloor \log n \rfloor \rfloor) = 1$.

\vspace{8pt}

\begin{center}
\footnotesize{
\begin{tikzpicture}[auto,domain=0:4,x=1cm,y=1cm] 
\pgftransformxscale{1.01}
\pgftransformyscale{.71}
\draw[red!60!white](0,10) -- (3.33,6.66);
\draw[red!60!white](3.33,6.66) -- (5,10);
\filldraw[red!60!white, fill=red!60!white] (5,10) -- (6,8) -- (7.5,10) -- (7.777,8.888) -- (8.75,10) -- (8.823,9.412) -- (9.375,10) --(5,10);
\draw[red!60!white, thick, dotted] (9.525,9.7) -- (9.85,9.7);

\draw[->] (-0.1,5) -- (10.4,5) node[right] {$\alpha$}; 
\draw[->] (0,6.366) -- (0,10.4) node[above] {$f(\alpha)$}; 
\draw[dashed] (0,5.5) -- (0,6.266);
\draw (0,5.1) -- (0,5.33);
\draw(0,4.5) node{$0$};

\draw[help lines] (3.333,6.366) -- (3.333,10);
\draw[help lines] (5,6.366) -- (5,10);
\draw[help lines] (6,6.366) -- (6,10);
\draw[help lines] (7.5,6.366) -- (7.5,10);
\draw[help lines] (7.777,6.366) -- (7.777,10);
\draw[help lines] (8.75,6.366) -- (8.75,10);
\draw[help lines] (8.823,6.366) -- (8.823,10);
\draw[help lines] (9.375,6.366) -- (9.375,10);
\draw[help lines] (10,6.366) -- (10,10);

\draw[help lines, dashed] (3.333,5.5) -- (3.333,6.266);
\draw[help lines, dashed] (5,5.5) -- (5,6.266);
\draw[help lines, dashed] (6,5.5) -- (6,6.266);
\draw[help lines, dashed] (7.5,5.5) -- (7.5,6.266);
\draw[help lines, dashed] (7.777,5.5) -- (7.777,6.266);
\draw[help lines, dashed] (8.75,5.5) -- (8.75,6.266);
\draw[help lines, dashed] (8.823,5.5) -- (8.823,6.266);
\draw[help lines, dashed] (9.375,5.5) -- (9.375,6.266);
\draw[help lines, dashed] (10,5.5) -- (10,6.266);
\draw[help lines] (3.333,5.1) -- (3.333,5.33);
\draw[help lines] (5,5.1) -- (5,5.33);
\draw[help lines] (6,5.1) -- (6,5.33);
\draw[help lines] (7.5,5.1) -- (7.5,5.33);
\draw[help lines] (7.777,5.1) -- (7.777,5.33);
\draw[help lines] (8.75,5.1) -- (8.75,5.33);
\draw[help lines] (8.823,5.1) -- (8.823,5.33);
\draw[help lines] (9.375,5.1) -- (9.375,5.33);
\draw[help lines] (10,5.1) -- (10,5.33);

\draw[help lines] (0,6.666) -- (10,6.666);
\draw[help lines] (0,8) -- (10,8);
\draw[help lines] (0,8.888) -- (10,8.888);
\draw[help lines] (0,9.412) -- (10,9.412);
\draw[help lines] (0,10) -- (10,10);

\draw(0,4.9) -- (0,5.1);
\draw(3.333,4.9) -- (3.333,5.1);
\draw(5,4.9) -- (5,5.1);
\draw(6,4.9) -- (6,5.1);
\draw(7.5,4.9) -- (7.5,5.1);
\draw(7.777,4.9) -- (7.777,5.1);
\draw(8.75,4.9) -- (8.75,5.1);
\draw(8.823,4.9) -- (8.823,5.1);
\draw(9.375,4.9) -- (9.375,5.1);
\draw(10,4.9) -- (10,5.1);
\draw(3.333,4.5) node{$\frac{1}{3}$};
\draw(5,4.5) node{$\frac{1}{2}$};
\draw(6,4.5) node{$\frac{3}{5}$};
\draw(7.5,4.5) node{$\frac{3}{4}$};
\draw(7.777,4.5) node{$\frac{7}{9}$};
\draw(8.63,4.5) node{$\frac{7}{8}$};\draw(8.93,4.5) node{$\frac{15}{17}$};
\draw(9.375,4.5) node{$\frac{15}{16}$};
\draw[dotted, thick] (9.525,4.5) -- (9.85,4.5);
\draw(10,4.5) node{$1$};

\draw(-0.1,6.666) -- (0.1,6.666);
\draw(-0.1,8.0) -- (0.1,8.0);
\draw(-0.1,8.888) -- (0.1,8.888);
\draw(-0.1,9.412) -- (0.1,9.412);
\draw(-0.1,10) -- (0.1,10);
\draw(-0.3,5) node{$0$};
\draw(-0.3,6.666) node{$\frac{2}{3}$};
\draw(-0.3,8.0) node{$\frac{4}{5}$};
\draw(-0.3,8.888) node{$\frac{8}{9}$};
\draw(-0.3,9.412) node{$\frac{16}{17}$};
\draw(-0.3,10) node{$1$};
\end{tikzpicture}}
\smallskip
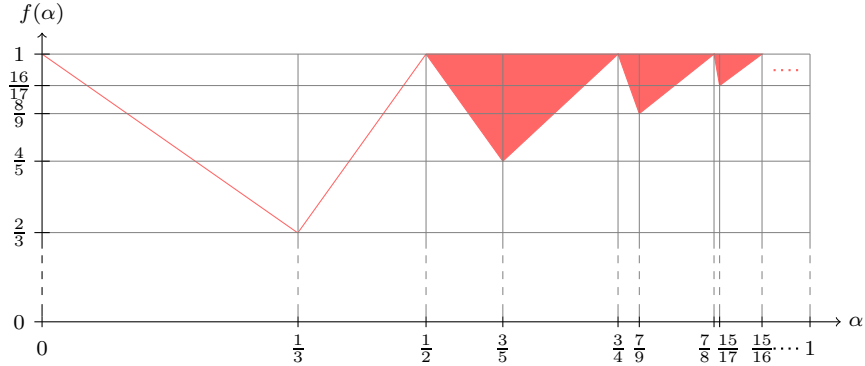
\\\parbox[h]{10.5cm}{\footnotesize Figure 1: Possible values of $f(\alpha)$ for $0 < \alpha < 1$. Any value of $f(\alpha)$ must be on the line (if $0 < \alpha < 1/2$) or in the filled area (if $1/2 < \alpha < 1$). As established in Theorem \ref{f(alpha)-thm}, $e_d(n)$ grows asymptotically at least as fast as $n^{2/3}$ and the minimal asymptotic growth is matched exactly for $d=n^{1/3}$.}
\end{center}


\vspace{10pt}

\begin{thebibliography}{10}

\bibitem{AlonPudlak}
Alon, N., Pudl\'{a}k, P.: Equilateral sets in $l^n_p$.
\newblock Geom.~Funct.~Anal. \textbf{13}(3), 467--482 (2003)

\bibitem{BandeltChepoiLaurent}
Bandelt, H.J., Chepoi, V., Laurent, M.: Embedding into rectilinear spaces.
\newblock Discrete~Comput.~Geom. \textbf{19}, 595--604 (1998)

\bibitem{BogdanovaetAl}
Bogdanova, G.T., Zinov'ev, V.A., Todorov, T.{\u{I}}.: On the construction of
  {$q$}-ary equidistant codes.
\newblock Problemy Peredachi Informatsii \textbf{43}(4), 13--36 (2007)

\bibitem{BoeseShrikhande}
Bose, R.C., Shrikhande, S.S.: A note on a result in the theory of code
  construction.
\newblock Information and Control \textbf{2}, 183--194 (1959)

\bibitem{BSSS1990}
Brouwer, A.E., Shearer, J.B., Sloane, N.J.A., Smith, W.D.: A new table of
  constant weight codes.
\newblock IEEE Trans. Inform. Theory \textbf{36}(6), 1334--1380 (1990)

\bibitem{Deza1973}
Deza, M.: Une propri\'et\'e extr\'emale des plans projectifs finis dans une
  classe de codes \'equidistants.
\newblock Discrete Math. \textbf{6}, 343--352 (1973)

\bibitem{FPRU90}
Feige, U., Peleg, D., Raghavan, P., Upfal, E.: Randomized broadcast in
  networks.
\newblock Random Structures Algorithms \textbf{1}(4), 447--460 (1990)

\bibitem{FKLW03}
Fu, F.W., Kl\o{}ve, T., Luo, Y., Wei, V.K.: On equidistant constant weight
  codes.
\newblock Discrete Appl. Math. \textbf{128}(1), 157--164 (2003)

\bibitem{Olla-Podrida}
Guy, R.K.: Unsolved {P}roblems: {A}n {O}lla-{P}odrida of {O}pen {P}roblems,
  {O}ften {O}ddly {P}osed.
\newblock Amer. Math. Monthly \textbf{90}(3), 196--200 (1983)

\bibitem{HC98}
Heng, I., Cooke, C.H.: Error correcting codes associated with complex
  {H}adamard matrices.
\newblock Appl. Math. Lett. \textbf{11}(4), 77--80 (1998)

\bibitem{HuffmanPless2003}
Huffman, W.C., Pless, V.: Fundamentals of error-correcting codes.
\newblock Cambridge University Press, Cambridge (2003)

\bibitem{KoolenLaurentSchrijver}
Koolen, J., Laurent, M., Schrijver, A.: Equilateral dimension of the
  rectilinear space.
\newblock Designs, Codes Cryptogr. \textbf{21}, 149--164 (2000)

\bibitem{vanLint1973}
van Lint, J.H.: A theorem on equidistant codes.
\newblock Discrete Math. \textbf{6}, 353--358 (1973)

\bibitem{MacWilliamsSloane1977}
MacWilliams, F.J., Sloane, N.J.A.: The theory of error-correcting codes. {I}.
\newblock North-Holland Publishing Co., Amsterdam (1977).
\newblock North-Holland Mathematical Library, Vol. 16

\bibitem{Swanepoel2004}
Swanepoel, K.J.: Equilateral sets in finite-dimensional normed spaces.
\newblock In: Seminar of {M}athematical {A}nalysis, \emph{Colecc. Abierta},
  vol.~71, pp. 195--237 (2004)

\end{thebibliography}

\end{document}